\documentclass[12pt]{article}
\usepackage[latin1]{inputenc}
\usepackage[british]{babel}
\usepackage{cmap}
\usepackage{lmodern}

\usepackage{amssymb, amsmath, amsthm}
\usepackage[a4paper,top=25mm,bottom=25mm,left=25mm,right=25mm]{geometry}
\usepackage{etex}
\usepackage{ragged2e}

\usepackage{authblk} 
\usepackage{pifont}
\usepackage{graphicx}
\usepackage[usenames,dvipsnames,svgnames,table]{xcolor}
\usepackage[figuresright]{rotating}
\usepackage{xtab} 
\usepackage{longtable} 
\usepackage{multirow}
\usepackage{footnote}
\usepackage[stable]{footmisc}
\usepackage{chngpage} 
\usepackage{pdflscape} 
\usepackage{tocbibind} 

\usepackage{pgfplots}
\pgfplotsset{compat=1.14}
\usepackage{setspace}

\makesavenoteenv{tabular}
\usepackage{tabularx}
\usepackage{booktabs}
\usepackage{threeparttable}
\usepackage[referable]{threeparttablex} 
\newcolumntype{R}{>{\raggedleft\arraybackslash}X}
\newcolumntype{L}{>{\raggedright\arraybackslash}X}
\newcolumntype{C}{>{\centering\arraybackslash}X}
\newcolumntype{A}{>{\columncolor{gray!25}}C}
\newcolumntype{a}{>{\columncolor{gray!25}}c}

\newlength{\tablen}

\usepackage{dcolumn} 
\newcolumntype{.}{D{.}{.}{-1}}

\usepackage{tikz}
\usetikzlibrary{arrows, calc, matrix, patterns, positioning, trees}
\usepackage[semicolon]{natbib}
\usepackage[hyphens]{url}
\usepackage{hyperref} 
\hypersetup{
  colorlinks   = true,    
  urlcolor     = blue,    
  linkcolor    = blue,    
  citecolor    = ForestGreen      
}
\usepackage{microtype}
\usepackage[justification=centerfirst]{caption} 

\usepackage[labelformat=simple]{subcaption}

\DeclareCaptionLabelFormat{parenthesis}{(#2)}
\makeatletter
\renewcommand\p@subfigure{\arabic{figure}.}
\makeatother

\DeclareCaptionLabelFormat{parenthesis}{(#2)}
\makeatletter
\renewcommand\p@subtable{\arabic{table}.}
\makeatother

\usepackage{enumitem}

\setlist[itemize]{leftmargin=2.5\parindent}
\setlist[enumerate]{leftmargin=2.5\parindent}

\theoremstyle{plain}

\newtheorem{proposition}{Proposition}[section]

\theoremstyle{definition}

\newtheorem{definition}{Definition}[section]
\newtheorem{example}{Example}[section]

\theoremstyle{remark}

\newtheorem{remark}{Remark}[section]

\def\keywords{\vspace{.5em} 
{\noindent \textit{Keywords}: }}

\def\JEL{\vspace{.5em} 
{\noindent \textbf{\emph{JEL} classification number}: }}

\def\AMS{\vspace{.5em} 
{\noindent \textbf{\emph{MSC} class}: }}

\author{\href{https://sites.google.com/view/laszlocsato}{L\'aszl\'o Csat\'o}\thanks{~E-mail: \emph{laszlo.csato@sztaki.hu}} }
\affil{Institute for Computer Science and Control (SZTAKI) \\
Laboratory on Engineering and Management Intelligence \\ Research Group of Operations Research and Decision Systems}
\affil{Corvinus University of Budapest (BCE) \\
Department of Operations Research and Actuarial Sciences}
\affil{Budapest, Hungary}
\title{When neither team wants to win: \\ a flaw of recent UEFA qualification rules}
\date{\today}

\def\Dedication{
{\noindent
$\mathfrak{Wer}$ $\mathfrak{\ddot{u}ber}$ $\mathfrak{dem}$ $\mathfrak{Unm\ddot{o}glichen}$ $\mathfrak{das}$ $\mathfrak{M\ddot{o}gliche}$ $\mathfrak{vers\ddot{a}umt}$, $\mathfrak{der}$ $\mathfrak{ist}$ $\mathfrak{ein}$ $\mathfrak{Tor}$.}\footnote{~``\emph{He who neglects the possible in quest of the impossible is a fool.}'' (Source: Carl von Clausewitz: \emph{On War}, Book 8, Chapter 9 -- Plan of War when the Destruction of the Enemy is the Object. Translated by Colonel James John Graham, London, N. Tr\"ubner, 1873. \url{http://clausewitz.com/readings/OnWar1873/TOC.htm})}
\vspace{0.25cm}

\flushright
\noindent (Carl von Clausewitz: \emph{Vom Kriege})
\vspace{1cm} 
\justify }

\begin{document}

\maketitle
\thispagestyle{empty}
\Dedication

\begin{abstract}
\noindent
Tanking, the act of deliberately dropping points or losing a game in order to gain some other advantage, is usually seen as being against the spirit of sports. It can be even more serious if playing a draw is a (weakly) dominant strategy for both teams in a match, since this may lead to collusion.
We show that such a situation occurred in a particular football match. As our generalisation reveals, the root of the problem resides in the incentive incompatibility of certain UEFA qualification rules.
The governing bodies of major sports should choose strategy-proof tournament designs because of several reasons. First, they may lead to the elimination of a third, innocent team. Second, incentive incompatible rules may discourage both teams from scoring goals, and the players could be interested in improving other match statistics than the number of goals.


\keywords{sports rules; football; incentive compatibility; fairness; match fixing}

\AMS{62F07, 91A80, 91B14}

\JEL{C44, D71, Z20}
\end{abstract}

\clearpage

\section{Introduction} \label{Sec1}

Any tournament design should provide contestants with the appropriate incentives \citep{Szymanski2003}. \citet{VaziriDabadghaoYihMorin2018} consider not rewarding teams for poor performance as an important aspect of the fairness of sports ranking methods.

However, ill-designed rules sometimes inspire \emph{tanking}, the act of deliberately dropping points or losing a game. For instance, if a team is already eliminated from the playoff of a league, it can be profitable to decrease the effort in order to secure a better position in the draft for the next season \citep{TaylorTrogdon2002, Lenten2016, Fornwagner2019, LentenSmithBoys2018}. In other cases, strategic manipulation can be applied to play against a preferred competitor in the later stage of a tournament \citep{Pauly2014, Vong2017}.
\citet{KendallLenten2017} offer the first comprehensive review of sports rules that have led to similar unexpected consequences.

The current paper deals with a particular case of tanking by assuming that a team exerts a lower effort only if it cannot be worse off -- in other words, it is extremely risk-averse in the choice of its strategy.
Examples include a famous football 
match played by \href{https://en.wikipedia.org/wiki/Barbados_4\%E2\%80\%932_Grenada_(1994_Caribbean_Cup_qualification)}{Barbados and Grenada in the 1994 Caribbean Cup qualification} \citep[Section~3.9.4]{KendallLenten2017}, qualifications for FIFA World Cups \citep{DagaevSonin2013, Csato2017d}, and UEFA club competitions \citep{DagaevSonin2018, Csato2019b}.

While it is usually unacceptable if a team can benefit from losing, some arguments decrease the importance of incentive incompatibility in practice.
First, the probability that such a situation occurs is often (very) low: otherwise, pure ``evolution'' would have selected out the unfair rule.
Second, one may think that both teams cannot be interested in tanking.

In the following, the second reasoning is shown to be invalid, at least in a particular tournament format.
As a motivation, we will present a real-world football match when playing a draw was the only opportunity to qualify for the two teams playing against each other.
After that, this bizarre situation will be generalised, which also provides solutions on how to prevent the possibility of such \emph{forced collusion} in the future. These are our main contributions.

Consequently, the paper is connected to another direction of the literature discussing (potential) examples of collusion. According to \citet{ElaadKrumerKantor2018}, in more corrupt countries, football teams in danger of relegation have a higher probability to win in the last game of the season and then reciprocate by losing in the following year.
The topic has also got media coverage recently because the planned format of the \href{https://en.wikipedia.org/wiki/2026_FIFA_World_Cup}{2026 FIFA World Cup} contains groups composed of three teams, which creates an ``ideal'' environment for match fixing \citep{Guyon2018b, GuyonMonkovic2018}. \citet{Guyon2019a} aims to quantify how often this is expected to happen and proposes alternative designs that would eliminate or strongly decrease the risk of collusion.

To illustrate the strength of the applied concept, consider the notorious football match called ``\href{https://en.wikipedia.org/wiki/Disgrace_of_Gij\%C3\%B3n}{Disgrace of Gij\'on}'' (its German name, ``\href{https://de.wikipedia.org/wiki/Nichtangriffspakt_von_Gij\%C3\%B3n}{Nichtangriffspakt von Gij\'on}'' [Non-aggression pact of Gij\'on] is perhaps even more expressive) \citep[Section~3.9.1]{KendallLenten2017}.
This was the final game of Group 2 in the \href{https://en.wikipedia.org/wiki/1982_FIFA_World_Cup}{1982 FIFA World Cup} between West Germany and Austria. Winning by one or two goals for West Germany would result in both teams qualifying at the expense of Algeria, while all other results would lead to the elimination of either West Germany or Austria.
After ten minutes West Germany scored, and the remaining 80 minutes was characterized by few serious attempts to attack.
However, the rules have not failed in the sense that the teams would not lose anything by kicking an additional goal.

The rest of the paper is organised as follows. Section~\ref{Sec2} discusses the particular match where the optimal strategy was forced collusion of the teams.
The presented scenario is generalised in Section~\ref{Sec3}, and Section~\ref{Sec4} concludes.

\section{A real-world example} \label{Sec2}

The \href{https://en.wikipedia.org/wiki/2016_UEFA_European_Under-17_Championship_qualification}{2016 UEFA European Under-17 Championship qualification} was a football competition, organised by the UEFA (Union of European Football Associations), to determine the national teams participating in the 2016 UEFA European Under-17 Championship final tournament. 
$53$ national teams entered the competition, which was played in two subsequent group stages between September 2015 and April 2016. The qualifying round was followed by the elite round, where 32 teams were drawn into eight groups of four teams each.
All matches were played as a part of a round-robin tournament in one of the countries in the group, with three points awarded for a win, one for a draw and none for a defeat \citep[Article~13]{UEFA2015b}. Tie-breaking rules in the groups are detailed in \citet[Article~14]{UEFA2015b}.

The eight group winners and the seven runners-up qualified for the final tournament.
The second-placed teams were ranked on the basis of the matches played against the group winners and third-placed teams according to the following criteria \citep[Article~15]{UEFA2015b}:
(1) higher number of points;
(2) superior goal difference;
(3) higher number of goals scored;
(4) lower disciplinary points;
(5) drawing of lots.

\begin{table}[t]
\centering
\captionsetup{justification=centering}
\caption{2016 UEFA European Under-17 Championship \\ qualifying competition -- Elite round, Group 6}
\label{Table1}

\begin{subtable}{\linewidth}
\centering
\caption{Match results}
\label{Table1a}
\rowcolors{1}{}{gray!20}
    \begin{tabularx}{0.9\linewidth}{lLLc} \toprule
    Date  & First team & Second team & Result \\ \hline \showrowcolors
    29 March 2016, 18:30 & France & Greece & 1-0 \\
    29 March 2016, 18:30 & Austria & Iceland & 0-0 \\ \hline
    31 March 2016, 18:30 & France & Iceland & 1-0 \\
    31 March 2016, 18:30 & Greece & Austria & 0-0 \\ \hline
    3 April 2016, 15:00 & Austria & France & to be played \\
    3 April 2016, 15:00 & Iceland & Greece & to be played \\ \bottomrule
    \end{tabularx}
\end{subtable}

\vspace{0.5cm}
\begin{subtable}{\linewidth}
\centering
\captionsetup{justification=centerfirst}
\caption{Standing after two matchdays \\
\footnotesize{Pos = Position; W = Won; D = Drawn; L = Loss; GF = Goals for; GA = Goals against; GD = Goal difference; Pts = Points. All teams have played 2 matches.}}
\label{Table1b}
\rowcolors{1}{}{gray!20}
    \begin{tabularx}{\linewidth}{Cl CCCC CC >{\bfseries}C} \toprule \showrowcolors
    Pos   & Team   & W     & D     & L     & GF    & GA    & GD    & Pts \\ \hline
    1     & France & 2     & 0     & 0     & 2     & 0     & 2     & 6 \\
    2     & Austria & 0     & 2     & 0     & 0     & 0     & 0     & 2 \\
    3     & Greece & 0     & 1     & 1     & 0     & 1     & -1    & 1 \\
    3     & Iceland & 0     & 1     & 1     & 0     & 1     & -1    & 1 \\ \bottomrule    
    \end{tabularx}
\end{subtable}
\end{table}

A bizarre situation occurred in Group 6 of the \href{https://en.wikipedia.org/wiki/2016_UEFA_European_Under-17_Championship_qualification#Elite_round}{elite round}, which deserves further investigation. On the morning of 3 April 2016, each team in the group had one more game to play as Table~\ref{Table1} shows. All other groups were finished by this date, the (provisional) worst runner-up being Poland from Group 7 with $1$ point, $1$ goal for and $2$ goals against after discarding its match played against the fourth-placed team of its group. 

We will focus on the match between Greece and Iceland. Both teams are eliminated if Austria defeats France in the parallel match.
If Austria does not win against France, then the winner of the match Greece vs Iceland would be the runner-up with $4$ points. However, it would have only $1$ point, $0$ goals for and $1$ goal against in the ranking of second-placed teams due to discarding the match played against the fourth-placed team of the group (the loser of Greece vs Iceland), which leads to elimination as the worst runner-up.
Consequently, if Austria plays at least a draw against France, then both Greece and Iceland are eliminated.

However, if Austria is defeated by France, Greece or Iceland may qualify by playing a draw with scoring at least one goal (in the case of 0-0, Poland is better among the second-placed teams): three teams, Austria, Greece, and Iceland would have $2$ points and the group tie-breaking rule provides that Greece and Iceland overtake Austria due to their higher number of head-to-head goals scored.
The runner-up would be determined by lower disciplinary points (red card = 3 points, yellow card = 1 point, expulsion for two yellow cards in one match = 3 points), or by a coin toss if disciplinary points are equal \citep[Article~14]{UEFA2015b}.
If Greece vs Iceland is 1-1, the runner-up has an identical goal record to Poland among the second-placed teams, so disciplinary points count, but a draw of at least 2-2 ensures the qualification of the runner-up ahead of Poland.

To summarise, the weakly dominant strategy of Greece and Iceland is to play a draw of at least 2-2, independently of the result of the other match taking place at the same time between Austria and France. Provided that it is reached -- certainly not a big challenge since both teams are interested in this outcome --, their objective function is to increase the number of disciplinary points of the \emph{other team}, that is, to force the players of the opponent to commit offences.
Alternatively, since no team had an ex-ante advantage in disciplinary points (Greece collected one yellow card against France and two against Austria, while Iceland collected one yellow card against Austria and two against France), they could have colluded to play a draw of 2-2 and wait patiently for the coin toss to grab the only chance of qualification.

This football match is mentioned in \citet[footnote~64]{KendallLenten2017} without further discussion. The authors refer to an online forum, available at \url{https://www.reddit.com/r/soccer/comments/4d5ki3/theres_a_bizarre_situation_in_u17_euro_2016/}, \linebreak where a user with the nickname \emph{Liverpool} describes the situation. However, it contains a mistake since a draw of 2-2 would not be followed by a penalty shootout.

\section{Generalisation} \label{Sec3}

In the following, we show that the possible occurrence of a situation when both teams playing a match are interested in tanking is a common feature of similar qualifiers.

Consider a \emph{group-based qualification system} $\mathcal{Q}$, where the participating teams are divided into $k \geq 2$ round-robin groups such that the number of teams in group $G_i$ is $n_i$, $i =1, 2, \dots ,k$.
Under any set of match results $\mathbf{R}$, the top $a_i \geq 0$ teams in every group $G_i$ directly qualify, while the next $b_i \geq 0$ teams -- whose set is denoted by $B_i(\mathbf{R})$ -- are compared in a so-called \emph{repechage group}. The remaining $n_i - a_i - b_i \geq 0$ teams of group $G_i$ are eliminated.

Let $B(\mathbf{R}) = \cup_{i=1}^k B_i(\mathbf{R})$ be the set of teams in the repechage group, where only the group matches played against the top $c_i$ teams are considered and $a_i + b_i \leq c_i \leq n_i$.

The teams are allocated into three distinct sets $T_1(\mathbf{R})$, $T_2(\mathbf{R})$, and $T_3(\mathbf{R})$ based on the set of match results $\mathbf{R}$, which correspond to the set of teams directly qualified (the probability of participation in the next round is $1$), advanced to the play-offs (the probability of participation in the next round is positive but less than $1$), and being eliminated (the probability of participation in the next round is $0$), respectively. Consequently, each team prefers being in $T_1(\mathbf{R})$ to being in $T_2(\mathbf{R}) \cup T_3(\mathbf{R})$, and they favour being in $T_2(\mathbf{R})$ over being in $T_3(\mathbf{R})$.
The allocation is assumed to be monotonic, that is, no team can achieve a better position in its group and in the repechage group by having less number of points or scoring fewer goals in any matches.
Finally, since the set of match results $\mathbf{R}$ uniquely determine the fate of the teams, they should contain all statistics used for breaking ties including yellow and red cards, the outcome of potential penalty shootouts, etc.

\begin{example} \label{Examp1}
In the elite round of the 2016 UEFA European Under-17 Championship qualification, there are $k = 8$ groups with $n_i = 4$ teams each such that $a_i = 1$, $b_i = 1$, and $c_i = 3$ for all $i = 1,2, \dots ,8$.
The repechage group is composed of the second-placed teams from all groups, where the first seven directly qualify and the last team is eliminated, therefore $\left| B(\mathbf{R}) \cap T_1(\mathbf{R}) \right| = 7$ and $\left| B(\mathbf{R}) \cap T_3(\mathbf{R}) \right| = 1$, independently of the match results $\mathbf{R}$.
Furthermore, $\left| T_1(\mathbf{R}) \right| = 15$ as the eight top teams also directly qualify, $\left| T_2(\mathbf{R}) \right| = 0$ because there are no play-offs, and $\left| T_3(\mathbf{R}) \right| = 17$.
\end{example}

Let $\mathcal{Q}$ be a group-based qualification system with the set of match results $\mathbf{R}$. For any team $x$, denote by $m(\mathbf{R},x)$ the integer for which $x \in T_{m(\mathbf{R},x)} (\mathbf{R})$.

Let $\mathcal{Q}$ be a group-based qualification system, $\mathcal{R}$ be a family of sets of match results, and $x,y$ be two teams in the same group $G_i$.
A family of sets of match results is considered because a given set of match results $\mathbf{R}$ completely determines the sets $T_1(\mathbf{R})$, $T_2(\mathbf{R})$, and $T_3(\mathbf{R})$, while we should show that both teams have a positive probability to qualify unless they kick additional goals.

For any set of match result $\mathbf{R} \in \mathcal{R}$, a set of match results $\mathbf{R_{xy}'}$ is said to be \emph{advantageous} for team $x$ with respect to its match(es) played against $y$ if $\mathbf{R}$ and $\mathbf{R_{xy}'}$ are identical but team $x$ has scored more goals against $y$ according to $\mathbf{R_{xy}'}$. This family of sets of match results is denoted by $\mathcal{R}_{xy}(\mathbf{R})$.


\begin{definition} \label{Def1}
\emph{Forced collusion}:
Let $\mathcal{Q}$ be a group-based qualification system, $\mathcal{R}$ be a family of sets of match results, and $x,y$ be two teams in the same group $G_i$.
Teams $x$ and $y$ are \emph{forced to collude} under $\mathcal{R}$ if:
\begin{itemize}
\item
for all set of match results $\mathbf{R} \in \mathcal{R}$, $\mathbf{R_{xy}'} \in \mathcal{R}_{xy}(\mathbf{R})$ implies $m(\mathbf{R_{xy}'},x) \geq m(\mathbf{R},x)$ and $\mathbf{R_{yx}'} \in \mathcal{R}_{yx}(\mathbf{R})$ implies $m(\mathbf{R_{yx}'},y) \geq m(\mathbf{R},y)$; and 
\item
there exists at least one set of match results $\mathbf{R} \in \mathcal{R}$ such that $\mathbf{R_{xy}'} \in \mathcal{R}_{xy}(\mathbf{R})$ implies $m(\mathbf{R_{xy}'},x) > m(\mathbf{R},x)$; and
\item
there exists at least one set of match results $\mathbf{R} \in \mathcal{R}$ such that $\mathbf{R_{yx}'} \in \mathcal{R}_{yx}(\mathbf{R})$ implies $m(\mathbf{R_{yx}'},y) > m(\mathbf{R},y)$.
\end{itemize}
\end{definition}

The first condition of Definition~\ref{Def1} means that teams $x$ and $y$ cannot benefit from scoring more goals with respect to qualification under any set of match results $\mathbf{R} \in \mathcal{R}$.
However, according to the second (third) condition, there is a possible scenario when team $x$ ($y$) would be worse off due to a better performance. Since these can be different situations, it would not be sufficient to base the definition only on one set of match results $\mathbf{R}$.
Thus forced collusion arises when the weakly dominant strategy of both teams $x$ and $y$ is to kick no more goals against each other.

\begin{example} \label{Examp2}
As we have seen in Section~\ref{Sec2}, the elite round of the 2016 UEFA European Under-17 Championship qualification allows for forced collusion under the family of sets of match results $\mathcal{R}$, which is given by the real match results except for the following:
\begin{itemize}
\item
France vs Austria is 1-0 for all $\mathbf{R} \in \mathcal{R}$;
\item
Iceland vs Greece is 2-2 for all $\mathbf{R} \in \mathcal{R}$;
\item
the number of yellow cards collected by the teams in the match Iceland vs Greece can be arbitrary in $\mathcal{R}$.
\end{itemize}
Let team $x$ be Iceland and team $y$ be Greece.
Now there is a set of match results $\mathbf{R} \in \mathcal{R}$ when $x \in T_1(\mathbf{R})$ and there is a set of match results $\mathbf{R} \in \mathcal{R}$ when $y \in T_1(\mathbf{R})$, depending on which team collects more disciplinary points.

However, if Iceland scores some additional goals against Greece in any $\mathbf{R} \in \mathcal{R}$, then $x \in T_3(\mathbf{R_{xy}'})$ for all $\mathbf{R_{xy}'} \in \mathcal{R}_{xy}(\mathbf{R})$ independently of $\mathbf{R} \in \mathcal{R}$, which is the most unfavourable outcome for Iceland.
Similarly, if Greece scores some additional goals against Iceland in any $\mathbf{R} \in \mathcal{R}$, then $y \in T_3(\mathbf{R_{yx}'})$ for all $\mathbf{R_{yx}'} \in \mathcal{R}_{yx}(\mathbf{R})$ independently of $\mathbf{R} \in \mathcal{R}$, which is the most unfavourable outcome for Greece.
\end{example}

The setting above may lead to forced collusion.

\begin{proposition} \label{Prop1}
Let $\mathcal{Q}$ be a group-based qualification system.
There exist teams $x,y$ from the same group $G_i$ that are forced to collude under the family of sets of match results $\mathcal{R}$ if the following two conditions hold simultaneously:
\begin{itemize}
\item
there is a difference in the allocation of teams in the repechage group, that is, at least two of the sets $B(\mathbf{\bar{R}}) \cap T_1(\mathbf{\bar{R}})$, $B(\mathbf{\bar{R}}) \cap T_2(\mathbf{\bar{R}})$, and $B(\mathbf{\bar{R}}) \cap T_3(\mathbf{\bar{R}})$ are non-empty under any set of match results $\mathbf{\bar{R}} \in \mathcal{R}$; and
\item
$b_i \geq 1$ and a proper subset of group matches played against lower ranked teams in group $G_i$ are considered in the repechage group, that is, $a_i + b_i < c_i < n_i$.
\end{itemize}
\end{proposition}

\begin{proof}
It is sufficient to provide an example where two teams are forced to collude.
Example~\ref{Examp2} can be used as a starting point with the family of sets of match results $\mathcal{R}$. Restrict the 2016 UEFA European Under-17 Championship qualifying competition to Groups 6 and 7. Let $G_i$ be Group 6. 
Note that $a_i = 1$, $b_i = 1$, $c_i = 3$, and $n_i = 4$, as well as $\left| B(\mathbf{\bar{R}}) \cap T_1(\mathbf{\bar{R}}) \right| = 1$ and $\left| B(\mathbf{\bar{R}}) \cap T_3(\mathbf{\bar{R}}) \right| = 1$ under any set of match results $\mathbf{\bar{R}} \in \mathcal{R}$.

Greece and Iceland are forced to collude. The construction works for $a_i = 0$, $b_i = 1$, $c_i = 2$, and $n_i = 3$ by leaving out France from Group 6.
$a_i$ and $n_i$ can be increased in parallel by adding a team with a copy of France which plays 0-0 against France.
$b_i$ and $n_i$ can be increased in parallel by adding a team with a copy of Greece/Iceland which plays 1-1 against both Greece and Iceland.
$c_i$ and/or $n_i$ can be increased by adding a team with a copy of Austria which plays 0-0 against Austria.

This example remains valid if $\left| B(\mathbf{\bar{R}}) \cap T_2(\mathbf{\bar{R}}) \right| = 1$ instead of $\left| B(\mathbf{\bar{R}}) \cap T_1(\mathbf{\bar{R}}) \right| = 1$, or $\left| B(\mathbf{\bar{R}}) \cap T_2(\mathbf{\bar{R}}) \right| = 1$ instead of $\left| B(\mathbf{\bar{R}}) \cap T_3(\mathbf{\bar{R}}) \right| = 1$ because these modifications do not affect the incentives of the teams.
It is also clear that the number of groups can be arbitrarily increased.
\end{proof}

According to the first requirement of Proposition~\ref{Prop1}, achieving a better position among the teams of the repechage group can be effective with respect to qualification. A crucial part of the problematic match in Section~\ref{Sec2} was the differentiation of runners-up since $\left| B(\mathbf{\bar{R}}) \cap T_1(\mathbf{\bar{R}}) \right| = 7$ and $\left| B(\mathbf{\bar{R}}) \cap T_3(\mathbf{\bar{R}}) \right| = 1$.

The second condition provides that a team can improve its position in the repechage group through not kicking more goals, which is impossible if all matches played against the teams ranked lower than the particular team that is relegated to the repechage group from $G_i$ are either counted ($c_i = n_i$) or discarded ($a_i + b_i = c_i$). The former is trivial. If all matches played against the teams ranked lower are discarded and $b_i = 1$, then influencing the ranking below the team that is relegated to the repechage group has no effect on the ranking in the repechage group, so scoring a goal cannot be unfavourable for this particular team.

Nonetheless, it is not a necessary condition of forced collusion because it might happen even if $a_i + b_i = c_i$ provided that $b_i \geq 2$: consider the case when both the second- and the third-placed teams of the groups are relegated to the repechage group in the example of Section~\ref{Sec2}. But we do not know any use of this policy in practice.

\begin{remark} \label{Rem1}
The elite round of the UEFA European Under-17 Championship qualification used the same format between \href{https://en.wikipedia.org/wiki/2016_UEFA_European_Under-17_Championship_qualification}{2016} and \href{https://en.wikipedia.org/wiki/2016_UEFA_European_Under-17_Championship_qualification}{2019}.
However, in the elite round of the \href{https://en.wikipedia.org/wiki/2020_UEFA_European_Under-17_Championship_qualification}{2020 UEFA European Under-17 Championship qualification}, scheduled to be played in spring 2020, all results are taken into account in the comparison of second-placed teams, which implies that forced collusion is no longer possible. This difference is noteworthy in the official regulations of the tournaments as \citet{UEFA2019c} essentially follows word by word \citet{UEFA2018f} with the exception of Article~15.02.

The first round (called qualifying round) of the \href{https://en.wikipedia.org/wiki/2020_UEFA_European_Under-17_Championship_qualification}{2020 UEFA European Under-17 Championship qualification} does not satisfy the requirements of Proposition~\ref{Prop1}, too, because the third-placed teams of the groups are compared on the basis of their matches played against the first- and second-placed teams in their groups, that is, $a_i + b_i = c_i$, therefore the second condition of Proposition~\ref{Prop1} does not hold.
\end{remark}


\begin{remark} \label{Rem2}
The root of the problem described in Section~\ref{Sec2} is clearly the comparison of the second-placed teams, which have not played any matches against each other, in a badly designed repechage group.
Some reasonable policies to avoid the negative result of Proposition~\ref{Prop1} are detailed in \citet[Section~5]{Csato2020f}. In the case of the elite round of the 2016 UEFA European Under-17 Championship qualification, since each group consists of four teams, the most straightforward solution is to consider all matches in the comparison of the second-placed teams, which has been implemented in the elite round of the \href{https://en.wikipedia.org/wiki/2020_UEFA_European_Under-17_Championship_qualification}{2020 UEFA European Under-17 Championship qualification} according to Remark~\ref{Rem1}.

However, UEFA has probably chosen to discard the match played against the last team of the group because of the different strength of the weakest teams. Thus our alternative proposal is to discard the match played against the team drawn from the last pot, which also guarantees incentive compatibility. Concerning the example of Section~\ref{Sec2}, this means the elimination of the match played against Iceland if the runner-up of the group is Austria, France, or Greece, and the match against Greece (the team drawn from the penultimate pot) if the runner-up is Iceland.
Consequently, both Greece and Iceland would lose the only chance of qualification before the last matchday, so they would not be forced to collude.

Nonetheless, this implies that the last two games in the group become unimportant with serious financial consequences. Incentive compatibility is obviously not the only aspect of evaluating sports rules, however, such trade-offs are beyond the scope of the current research.
\end{remark}

\begin{remark} \label{Rem3}
If there is a group-based qualification system allowing for forced collusion of two teams, and tie-breaking rules contain other match statistics than the number of goals (such as disciplinary points), then the teams may want to influence this particular match statistics rather than the number of goals.
\end{remark}

It is worth noting that UEFA modified tie-breaking rules from the \href{https://en.wikipedia.org/wiki/2017_UEFA_European_Under-17_Championship}{2017 UEFA European Under-17 Championship} and its \href{https://en.wikipedia.org/wiki/2017_UEFA_European_Under-17_Championship_qualification}{qualification} \citep{UEFA2016e}, as well as from the \href{https://en.wikipedia.org/wiki/2017_UEFA_European_Under-19_Championship}{2017 UEFA European Under-19 Championship} and its \href{https://en.wikipedia.org/wiki/2017_UEFA_European_Under-19_Championship_qualification}{qualification}. To be more specific, they have added a criterion before drawing of lots both in the groups and in the repechage group as a higher position in the coefficient ranking list used for the qualifying round draw. While this still does not give a strict ranking of the teams, the reform substantially decreases the probability that some teams should be ranked randomly, by a coin toss.

However, the change might have worsened the situation described in Section~\ref{Sec2}. The coefficient of Greece was lower than the coefficient of Iceland, therefore the only opportunity for Greece to qualify would be to play a draw of at least 2-2 and ensure that Iceland has more disciplinary points, that is, the players of Iceland get more yellow/red cards during the match.
This strange incentive is probably against the intentions of the administrators as the primary goal of every team should be to score as many goals as possible.

\section{Conclusions} \label{Sec4}

Because Austria defeated France by 2-1, Greece and Iceland had no chance to qualify (the result was 0-1). Perhaps the players of the latter teams were well-informed on the parallel match, and they did not want to risk the collusion.

Nevertheless, the (weakly) dominant strategy of Greece and Iceland seems to be against the intention of organisers. First, it may lead to the elimination of a third, innocent team. Second, the rules may discourage both teams from scoring goals, furthermore, the players can be interested in improving other match statistics than the number of goals.

This problematic situation was barely avoided during the 2016 UEFA European Under-17 Championship qualifying competition. It should be an undeniable warning for the governing bodies of major sports to apply strategy-proof tournament designs in the future. UEFA has recently changed its policy in this direction.

\section*{Acknowledgements}
\addcontentsline{toc}{section}{Acknowledgements}
\noindent
We are grateful to \emph{Dmitry Dagaev}, \emph{Julien Guyon}, and \emph{Alex Krumer} for beneficial remarks. \\
An anonymous reviewer and \emph{Tam\'as Halm} provided valuable comments and suggestions on an earlier draft. \\
We are indebted to the \href{https://en.wikipedia.org/wiki/Wikipedia_community}{Wikipedia community} for contributing to our research by collecting and structuring basic information on the tournaments discussed in the paper. \\
The research was supported by the MTA Premium Postdoctoral Research Program grant PPD2019-9/2019.

\bibliographystyle{apalike}
\bibliography{All_references}

\end{document}